\documentclass[a4paper,10pt]{article}%

\usepackage[numbers,sort&compress]{natbib}

\usepackage{booktabs,multirow,array}
\usepackage{enumitem}
\setlist[itemize,1]{nosep}
\usepackage{amsmath}
\usepackage{amsfonts}
\usepackage{amssymb}
\usepackage{amsthm}
\usepackage{graphicx}
\usepackage{xcolor}
\usepackage{tikz}
\usetikzlibrary{shadows,arrows}
\usepackage{listings}

\usepackage{thmtools}
\usepackage{thm-restate}

\renewcommand{\mid}{:}

\makeatletter
\renewcommand\bibsection%
{
  \section*{\refname
    \@mkboth{\MakeUppercase{\refname}}{\MakeUppercase{\refname}}}
}
\makeatother
\newtheorem{theorem}{Theorem}
\newtheorem{corollary}[theorem]{Corollary}
\newtheorem{lemma}[theorem]{Lemma}
\newtheorem{proposition}[theorem]{Proposition}
\newtheorem{definition}[theorem]{Definition}

\DeclareMathOperator*{\argmax}{\arg\!\max}
\newcommand{\R}{\mathbb{R}}
\renewcommand{\P}{\mathcal{P}}

\title{Sharing Non-Anonymous Costs of Multiple Resources Optimally}
\author{Max Klimm\thanks{Department of Mathematics, Technische Universit\"at Berlin, Stra\ss e des 17.~Juni 136, 10623~Berlin, Germany, \texttt{klimm@math.tu-berlin.de}.} \and Daniel Schmand\thanks{RWTH Aachen University, Kackertstra\ss e 7, 52072~Aachen, Germany, \texttt{daniel.schmand@oms.rwth-aachen.de}. Most of the work was done while this author was at Technische Universit\"at Berlin.}}

\begin{document}
\maketitle

\begin{abstract}
In cost sharing games, the existence and efficiency
of pure Nash equilibria fundamentally depends on the method that is used to share the resources' costs. We consider a general class of resource allocation problems in which
a set of resources is used by a heterogeneous set of selfish users. The cost of a resource is a (non-decreasing) function of the \emph{set} of its users.
 %
Under the assumption that the costs of the resources are shared by \emph{uniform} cost sharing protocols, i.e., protocols that use only local information
of the resource's cost structure and its users to determine the cost shares, we exactly quantify the inefficiency of the resulting
pure Nash equilibria. Specifically, we show tight bounds on prices of stability and anarchy for games with only submodular and only
supermodular cost functions, respectively, and an asymptotically tight bound for games with arbitrary set-functions. While all our upper bounds are attained for the
well-known Shapley cost sharing protocol, our lower bounds hold for arbitrary uniform cost sharing protocols and are even valid for games with anonymous costs, i.e.,
games in which the cost of each resource only depends on the cardinality of the set of its users. 
\end{abstract}

\section{Introduction}
Resource allocation problems are omnipresent in many areas of economics, computer science, and operations research with many applications,
e.g., in routing, network design, and scheduling. Roughly speaking, when solving these problems the central question is how to allocate a given set of resources
to a set of potential users so as to optimize a given social welfare function. In many applications, a major issue is that the users of the system are striving
to optimize their own private objective instead of the overall performance of the system. Such systems of users with different
objectives are analyzed using the theory of non-cooperative games.

A fundamental model of resource allocation problems with selfish users are congestion games \citep{Rosenthal73a}.
In a congestion game, each user chooses a subset of a given set of resources out of a set of allowable subsets.
The cost of each resource depends on the number of players using the particular resource. The private cost of each user equals the sum of costs given by the cost functions of the
resources contained in the chosen subset. Congestion games can be interpreted as cost sharing games with fair cost allocation in which the cost of a resource is a function of the
number of its users and each user pays the same cost share,
see also \citet{Anshelevich08}. \citeauthor{Rosenthal73a} showed that every congestion game has a pure Nash equilibrium, i.e., a strategy vector such that no player
can decrease its cost by a unilateral change of her strategy. 

This existence result depends severely on the assumption that the players are identical in the sense that they contribute
in equal terms to the congestion---and, thus, the cost---on the resources, which is unrealistic in many applications. As a better model for
heterogeneous users, \citet{Milchtaich96} introduced \emph{weighted} congestion games, where each player is associated
with a positive weight and resource costs are functions of the aggregated weights of their respective users. It is well known that
these games may lack a pure Nash equilibrium \citep{Fotakis05,Goemans05,Libman01}. A further generalization of weighted congestion games with even more modeling
power are congestion games with set-dependent cost functions introduced by \citet{Fabrikant04}. Here, the cost of each resource is a (usually non-decreasing) function of
the \emph{set} of its users. Set-dependent cost function can be used to model multi-dimensional cost structures on the resources that may arise form different
technologies at the resources required by different users, such as bandwidth, personal, machines, etc. 

For the class of congestion games with set-dependent costs,~\citet{Gopalakrishnan13} precisely characterized how the resources' costs can be distributed among its users such that the existence of a pure
Nash equilibrium is guaranteed. Specifically, they showed that the class of generalized weighted Shapley protocols is the unique maximal class of cost sharing
protocols that guarantees the existence of a pure Nash equilibrium in all induced cost sharing games. This class of protocols is quite rich as it contains,
e.g., weighted versions of the Shapley protocol and ordered protocols as considered by \citet{ChenRV10}.
\citeauthor{ChenRV10} examined which protocols give rise to good equilibria for cost sharing games where each resource has a fixed cost that has to be paid if the
resource is used by at least one player. In subsequent work, von Falkenhausen and Harks~\cite{Falkenhausen13} gave a complete characterization of the prices of anarchy and stability that is achievable by cost
sharing protocols in games with weighted players and matroid strategy spaces. In this paper, we follow this line of research asking which cost sharing protocols
guarantee the existence of good equilibria for \emph{arbitrary} (non-decreasing) set-dependent cost functions and \emph{arbitrary} strategy spaces.
\paragraph{Our Results.}
We study cost sharing protocols for a general resource allocation model with set-dependent costs that guarantee the existence of efficient pure Nash equilibria.
Our results are summarized in Table~\ref{tab:results}. Specifically, we give tight and asymptotically tight bounds on the inefficiency of Nash equilibria
in games that use the Shapley protocol both in terms of the price of anarchy and the price of stability and for games with submodular, supermodular and arbitrary
non-decreasing cost functions, respectively. 
The lower bounds that we provide hold for \emph{arbitrary} uniform cost sharing protocols and even in network games with anonymous costs
in which the cost of a resource depends only of the cardinality of the set of its users. Our upper and lower bounds are exactly matching except for the price
of stability in games with arbitrary non-decreasing cost functions where they only match asymptotically. Nonetheless, the lower bound of $\Theta(nH_n) = \Theta(n \log n)$
for any uniform protocol is our technically most challenging result and relies on a combination of characterizations of stable protocols
for constant cost functions taken from \cite{ChenRV10} and set-dependent cost functions from \cite{Gopalakrishnan13}. Our results imply
that both for submodular and supermodular costs there is no
other uniform protocol that gives rise to better pure Nash equilibria than the Shapley protocol, in the worst case. As another interesting
side-product of our results, we obtain that moving from anonymous costs (that depend only on the number of users) to set-dependent costs
does not deteriorate the quality of pure Nash equilibria in cost sharing games.
\newcommand{\scs}{\scriptsize}
\begin{table}[tb]
\caption{\label{tab:results}The inefficiency of $n$-player cost sharing games.}
\begin{center}
\begin{tabular*}{\linewidth}{
@{}l
@{\extracolsep{\fill}} c
@{\extracolsep{2ex}}   c
@{\extracolsep{2ex}}   c
@{\extracolsep{5ex}}   c
@{\extracolsep{2ex}}   c
@{\extracolsep{2ex}}   c
}
\toprule
\multirow{2}{*}{Cost Functions} & \multicolumn{3}{c}{Price of stability} & \multicolumn{3}{c}{Price of anarchy}\\
& \scs Value & \scs Lower bd. & \scs Upper bd. & \scs Value & \scs Lower bd. & \scs Upper bd.\\
\midrule
submodular & $H_n$ & \scs\citep{Anshelevich08} & \scs Thm.~\ref{thm_pos_shapley_setDependent_submod_upper} & $n$ &
\scs\citep{ChenRV10} & \scs Thm.~\ref{thm_poa_shapley_setDependent_submod_upper}\\
supermodular & $n$ & \scs Pro.~\ref{pro_uniform_supermod_lower} & \scs
Thm.~\ref{thm_pos_shapley_setDependent_supermod_upper} & $\infty$ &
\scs\citep{Falkenhausen13}, Thm.~\ref{thm_uniform_poa_supermod_lower} & --\\
arbitrary & $\Theta(nH_n)$ & \scs Thm.~\ref{thm_uniform_arbitrary_lower} & \scs
Thm.~\ref{thm_pos_shapley_setDependent_arb_upper} & $\infty$ &
\scs\citep{Falkenhausen13}, Thm.~\ref{thm_uniform_poa_supermod_lower} & --\\
\bottomrule
\end{tabular*}
\end{center}
\end{table}
\paragraph{Related Work.}
To measure the inefficiency of equilibria, two notions have
evolved. The price of anarchy \cite{Koutsoupias99,Papadimitriou01} is the worst case ratio of the social
cost of an equilibrium and that of a social optimum. The price of stability \cite{Schulz03,Anshelevich08} is
the ratio of the social cost of the most favorable equilibrium and that of the social optimum.

\citet{ChenRV10} initiated the study of cost sharing protocols that guarantee the existence of efficient pure Nash
equilibria. They considered the case of constant resource costs and characterized the set of linear uniform protocols. For uniform protocols (that solely depend on local information) they showed that proportional cost sharing
guarantees a price of stability equal to the $n$-th harmonic number and a price of anarchy of $n$ for $n$-player games.
Von Falkenhausen and Harks~\cite{Falkenhausen13} studied cost sharing protocols for weighted congestion games with matroid strategy spaces. They gave various tight
bounds for the prices of anarchy and stability that is achievable by uniform and more general protocols. Among other
results, they showed that even for a bounded number of players, no uniform protocol can archive a constant price of
anarchy. \citet{Kollias11} studied weighted congestion games with arbitrary strategy spaces. They recalled a result from
\citet{Hart89} to deduce that every weighted congestion game has a pure Nash equilibrium if the cost of each resource is
distributed according to the (weighted) Shapley value. Furthermore, bounds on the price of stability for this cost sharing
protocol are given. \citet{Gopalakrishnan13} considered congestion games with arbitrary set-dependent cost functions. They give a
full characterization of the cost sharing protocols that guarantee the existence of a pure Nash equilibrium in all such
games. 

Independently of our work, \citet{Roughgarden14} showed that the price of stability of Shapley cost sharing games is $H_n$ provided
that the cost of each resource is a submodular function of its users. They further showed that the strong price of anarchy is $H_n$ as
well. In forthcoming work, \citet{Gkatzelis14} examined optimal cost sharing rules for weighted congestion games with polynomial cost
functions. In particular, they showed that among the set of weighted Shapley cost sharing methods, the best price of anarchy can be
guaranteed by unweighted Shapley cost sharing. Although in a similar vein, their result is independent from ours since our results
hold even for unweighted players and arbitrary (submodular, supermodular, or arbitrarily non-decreasing) costs, while their result
holds for weighted players and convex and polynomial costs. Yet, we believe that the combination of the results of \citeauthor{Gkatzelis14}
and ours give strong evidence that in multiple scenarios of interest, Shapley cost sharing is the best way to share the cost among selfish
users. Both papers thus contribute to the \emph{quantitative} justification of Shapley cost sharing, as opposed to the \emph{axiomatic}
justification originally proposed by \citet{Shapley53}.

\section{Preliminaries}
We are given a finite set of \emph{players} $N = \{1,\dots,n\}$ and a finite and non-empty set of \emph{resources} $R$. For each player
$i$, the set of \emph{strategies} $\mathcal{P}_i$ is a non-empty set of subsets of $R$. We set $\P = \P_1 \times \dots
\times \mathcal{P}_n$ and call $P=(P_1,\dots,P_n) \in \P$ a \emph{strategy vector}. For $P \in \P$ and $r \in R$, let
$P^r =\{i \mid r\in P_i\}$ denote the set of players that use resource $r$ in strategy vector~$P$. For each resource~$r$, we are given a non-negative
cost function $C^r: 2^N \rightarrow \mathbb{R}_{\geq 0}$ mapping the set of its users to the respective cost value.
We assume that all cost functions~$C^r$ are non-decreasing, in the sense that $C^r(T) \leq C^r(U)$ for all $T \subseteq U$ with
$T,U \in 2^N$, and that $C^r(\emptyset) = 0$. The function $C^r$ is \emph{submodular} if $C^r(X \cup \{i\}) - C^r(X) \geq C^r(Y \cup \{i\}) - C^r(Y)$
for all $X \subseteq Y \subseteq N$, $i \in N \setminus Y$ and \emph{supermodular} if  $C^r(X \cup \{i\}) - C^r(X) \leq C^r(Y \cup \{i\}) - C^r(Y)$ for
all $X \subseteq Y \subseteq N$, $i \in N \setminus Y$. We call $C^r$ \emph{anonymous} if $C^r(X) = C^r(Y)$ for all $X,Y \subseteq N$ with $|X|=|Y|$.
For anonymous cost functions, submodularity and supermodularity are equivalent to concavity and convexity, respectively.

The tuple $\mathcal{M} = (N,R,\P,(C^r)_{r \in R})$ is called a \emph{resource allocation model}. For the special case that $R$ corresponds to the
set of edges of a graph and for every~$i$ the set $\P_i$ corresponds to the set of $(s_i,t_i)$-paths for two designated vertices $s_i$ and $t_i$,
we call $\mathcal{M}$ a \emph{network} resource allocation model.
 
The players' private costs are governed by local cost sharing protocols that decide how the cost of each resource is divided among its users.
%
%
A cost sharing protocol is a family of functions $(c_i^r)_{i \in N, r \in R} : \mathcal{P} \to \R$ that determines for each resource~$r$ and each
player~$i$, the cost share $c_i^r(P)$ that player~$i$ has to pay for using resource $r$ under strategy vector $P$. The private cost of player~$i$
under strategy vector $P$ is then defined as $C_i(P) = \sum_{r \in
P_i}{c_i^r(P)}$.

A resource allocation model $\mathcal{M} = (N,R,\mathcal{P},(C^r)_{r \in R})$ together with a cost sharing protocol $(c^r_i)_{i \in N,r \in R}$
thus defines a strategic game $G$ with player set $N$, strategy space $\mathcal{P}$ and private cost functions $(C_i)_{i \in N}$. For a strategic
game $G$, let $\text{PNE} \subseteq \P$ denote the set of pure Nash equilibria of $G$. The price of anarchy of $G$ is then defined as
$\smash{\max_{P \in \text{PNE}(G)} C(P)/C(\hat{P})}$, and the price of stability is defined as $\smash{\min_{P \in \text{PNE}(G)} C(P)/C(\hat{P})}$,
where $C(P) = \sum_{r\in R}{C^r(P^r)}$ is the \emph{social cost} of a strategy vector $P$ and $\smash{\hat{P}}$ is a strategy vector that minimizes
$C$. The strategy vector $\hat{P}$ is called \emph{socially optimal}.

Throughout this paper, we impose the following
assumptions on the cost sharing protocol that have become standard in the literature, cf. \citep{ChenRV10,Falkenhausen13}:
\begin{itemize}[label=--]
\item \emph{Stability:} There exists at least one pure Nash equilibrium, i.e, there is $P \in \P$ such that $C_i(P) \leq C_i(\tilde{P}_i,P_{-i})$ for
all $i \in N$ and $\tilde{P}_i \in \P_i$. 
\item \emph{Budget-balance:} The cost of all used resources is exactly covered by the cost shares of its users, i.e.,
$\sum_{i \in P^r} c_i^r(P) = C^r(P^r)$ and $c_i^r(P) = 0$ for all $i \notin P^r$ for all $r \in R$ and $P \in
\mathcal{P}$.
\item \emph{Uniformity:} The cost shares depend only on the cost structure of the resource and the set of users, i.e., $c_i^r(P) = c_i^{\tilde{r}}(\tilde{P})$ for all $i \in N$ and all
resource allocation models $(N,\P,(C^r)_{r \in R})$, $(N,\tilde{\P},(\tilde{C}^{\tilde{r}})_{\tilde{r} \in
\tilde{R}})$ with ${C}^r \equiv {\tilde{C}}^{\tilde{r}}$ and all strategy vectors $P \in \P$, $\tilde{P} \in \tilde{\P}$ with $P^r = \tilde{P}^r$.
\end{itemize}
With some abuse of notation, we henceforth write $c_i^r(P^r)$ instead of $c_i^r(P)$.

In this work, we only consider cost sharing protocols that satisfy these assumptions but it is worth noting that, except for the lower bound of $\Theta(nH_n)$ on
the price of stability for arbitrary non-decreasing cost functions, all our results are valid for a slightly weaker notion of uniformity where the cost shares may
depend on the identity of the resource.

A protocol that is stable, budget-balanced and uniform is the \emph{Shapley cost sharing} protocol. To give a formal definition,
for a set $S \subseteq N$ of players, let us denote by $\Pi(S)$ the set of permutations $\pi : S \to \{1,\dots,|S|\}$. Let $\pi \in \Pi(P^r)$ be a
permutation of the players in $P^r$ and let $P^r_{i,\pi} = \{j \in P^r \mid \pi(j) < \pi(i)\}$ be the set of players that precede player~$i$ in $\pi$.
The \emph{Shapley value} of player $i$ at resource $r$ with users $P^r$ is then defined as
\begin{align}
\label{eq:shapley}
\phi_{i}^{r}(P^r) =
\begin{cases}
 \frac{1}{|P^r|!}\sum\limits_{\pi \in \Pi(P^r)}{C^r(P^r_{i,\pi}\cup \{i\}) - C^r(P^r_{i,\pi})}, & \text{if } i \in P^r,\\
 0 & \text{otherwise,}
\end{cases}
\end{align}
i.e., the Shapley cost share is the marginal increase in the cost due to player~$i$ averaged over all possible permutations of the
players in $P^r$.

The following proposition is an immediate consequence of \eqref{eq:shapley} and well known in the literature. For the sake of completeness, we give a proof here.

\begin{restatable}{proposition}{TrivialPropShapley}
\label{prop_shapley_budgetbalance_separable}
 The Shapley cost sharing protocol is budget-balanced and uniform.
\end{restatable}

\begin{proof}
Uniformity follows directly from \eqref{eq:shapley}, since the cost functions only depend on $P^r$ and $C^r$. For budget-balance, it is easy to see that
 \begin{align*}
  \sum\limits_{i \in P^r}{\phi_i^r(P^r)} 
  &= \sum\limits_{i \in P^r}{\frac{1}{|P^r|!}\sum\limits_{\pi \in \Pi(P^r)}{C^r(P^r_{i,\pi} \cup \{i\}) - C^r(P^r_{i,\pi})}}\\
  &= \frac{1}{|P^r|!}\sum\limits_{\pi \in \Pi(P^r)}{\sum\limits_{i \in P^r}{C^r(P^r_{i,\pi}\cup \{i\}) - C^r(P^r_{i,\pi})}}\\
  &= C^r(P^r),
 \end{align*}
 which gives the claimed result.
\end{proof}

To show that Shapley cost sharing is also stable, we follow the road taken by \citet{Kollias11},
who gave a potential function for Shapley cost sharing for weighted congestion games. The following result can be proven simply by verifying that each
step in the proof by \citeauthor{Kollias11} applies to set-dependent costs, too.
For the sake of completeness, we give a proof here.

\begin{restatable}{proposition}{propphiExactPot}
\label{prop_phiExactPot}
 Let $\pi \in \Pi(N)$ be arbitrary and let $\Phi : \mathcal{P} \rightarrow \mathbb{R}$ with $P \mapsto \sum_{r\in R}{\Phi^r(P)}$ and $\Phi^r(P) = \sum_{i \in P^r}{\phi_i^r(P^r_{i,\pi} \cup \{i\})}$.
 Then, $\Phi$ is an exact potential function for Shapley cost sharing games with set-dependent costs.
\end{restatable}

\begin{proof}
We first show that the resource potential $\Phi^r(P)$ is independent of the chosen permutation $\pi$. Let us fix $r \in R$ and $\pi \in \Pi(P^r)$. We obtain
 \begin{align}
  \Phi^r(P) &= \sum_{i \in P^r}{\phi_i^r(P^r_{i,\pi} \cup \{i\})}\nonumber \\
  &= \sum_{i \in P^r}{\frac{1}{|P^r_{i,\pi} \cup \{i\}|!}\sum\limits_{p \in \Pi(P^r_{i,\pi} \cup \{i\})}}{\left(C^r(P^r_{i,p}\cup \{i\}) - C^r(P^r_{i,p})\right)}\nonumber \\
  &= \sum_{i \in P^r} \sum_{p \in \Pi(P^r_{i,\pi} \cup \{i\})} \frac{C^r(P^r_{i,p} \cup \{i\}) - C^r(P^r_{i,p})}{|P^r_{i,\pi} \cup \{i\}|!}.
\label{eq_indep_perm}
\end{align}
We observe that \eqref{eq_indep_perm} is a weighted sum of $C^r(T)$ for some sets
 $T \subseteq P^r$. We set
\begin{align*}
\sum_{i \in P^r} \sum_{p \in \Pi(P^r_{i,\pi} \cup \{i\})} \frac{C^r(P^r_{i,p} \cup \{i\}) - C^r(P^r_{i,p})}{|P^r_{i,\pi} \cup \{i\}|!}
 = \sum_{T \subseteq P^r} \alpha_T \cdot C^r(T)
\end{align*}
and proceed to calculate the coefficients $\alpha_T$, $T \subseteq P^r$. To this end, for a fixed set $T$,
 we have to count how often $C^r(T)$ is in the sum of \eqref{eq_indep_perm}. Now let player $t \in T$ be the last player
 of $T$ in $\pi$ and let $l = \pi(t)$. We know that there are only positive contributions to $\alpha_T$ from
 the Shapley value of player $t$ and a permutation $p$, where the players in $T\setminus \{ t\}$ come first,
 are followed by $t$ and all the other players come after $t$. There are $(|T|-1)!(l-|T|)!$ many of these permutations
 and there are $l!$ permutations in the Shapley value of player $t$ in total. That means we have a positive contribution to
$\alpha_T$
 with amount
 \begin{align}
  \label{term_posConti}
   \frac{(|T|-1)!(l-|T|)!}{l!}\text{.}
 \end{align} 
 Now we consider the negative contribution. There is a negative contribution to $\alpha_T$ from the Shapley value of all
players $\pi^{-1}(j)$
 that come after $t$ in $\pi_r$, i.e., $j \in \{l+1,\dots,|P^r|\}$. For a fixed $j$ we have to count the permutations where
all the players of $T$ come first and
 are followed by $\pi^{-1}(j)$ immediately. We know that there are $|T|!(j-|T|-1)!$ of these permutations in total, so the
total negative  contribution for all players $\pi^{-1}(j)$ is
 \begin{align}
  \label{term_negContri}
  \sum_{j=l+1}^{|P^r|}{\frac{|T|!(j-|T|-1)!}{j!}}\text{.}
 \end{align}
 Now we can calculate $\alpha_T$ explicitly by subtracting \eqref{term_negContri} from \eqref{term_posConti}:
 \begin{align}\label{eq_alpha_T}
    \alpha_T &= \frac{(|T|-1)!(l-|T|)!}{l!} - \sum_{j=l+1}^{|P^r|}{\frac{|T|!(j-|T|-1)!}{j!}}\notag\\
    &= (|T|\!-\!1)!\Bigg(\frac{(l-|T|)!}{l!} - \sum_{j=l+1}^{|P^r|}{\frac{|T|(j-|T|-1)!}{j!}}\Bigg)\notag\\
    &= (|T|\!-\!1)!\Big(\frac{(l-|T|)!}{l!} - \sum_{j=l+1}^{|P^r|}{\left(\frac{(j-|T|-1)!}{(j-1)!} - \frac{(j-|T|)!}{j!}\right)}\Big)\notag\\
    &= (|T|\!-\!1)!\Big(\frac{(l-|T|)!}{l!} - \frac{(l+1-|T|-1)!}{(l+1-1)!} + \frac{(|P^r|-|T|)!}{|P^r|!}\Big)\notag\\
    &= (|T|\!-\!1)!\Big(\frac{(|P^r|-|T|)!}{|P^r|!}\Big).
\end{align}
We derive that $\alpha_T$ is independent of $l$ and, thus, independent of the permutation $\pi$.

To finish the proof, let $P \in \P$, $i \in N$, and $\tilde{P}_i \in \P_i$ be arbitrary. We have shown that when calculating
$\Phi = \sum_{r \in R} \sum_{i \in P^r}{\phi_i^r(P^r_{i,\pi} \cup \{i\})}$ it is without loss of generality to assume that player~$i$
appears last in $\pi$. We then obtain,

 \begin{equation*}
  \begin{split}
   \Phi(P) - \Phi(P_{-i},P_i') &= \sum_{r \in P_i\setminus P_i'}{\phi_i^r(P^r)} - \sum_{r' \in P_i'\setminus P_i}{\phi_{i}^{r'}(P^{r'}\cup \{i\})}\\
   &= C_i(P) - C_i(P_{-i},P_i'),
  \end{split}
 \end{equation*}
which finishes the proof.
\end{proof}

Using that potential games always have a pure Nash equilibrium, we obtain the following immediate corollary.

\begin{corollary}
\label{thm_shapley_stable}
 The Shapley cost sharing protocol is stable.
\end{corollary}

As a side-product of the proof of Proposition~\ref{prop_phiExactPot}, we obtain the following alternative representation of the exact potential function of a Shapley cost sharing game.

\begin{corollary}
\label{cor:alpha_T}
The exact potential function for Shapley cost sharing games can be written as $\Phi(P) = \sum_{r \in R} \sum_{T \subseteq P^r} \alpha_T \cdot C^r(T)$ where $\alpha_T = \frac{(|P^r|-|T|)! \cdot (|T|-1)!}{|P^r|!}$,
and $\alpha_{\emptyset}=0$.
\end{corollary}

\section{The Efficiency of Shapley Cost Sharing}
\label{sec:shapley}

Having established the existence of pure Nash equilibria in Shapley cost sharing games we proceed to analyze their efficiency. We start to consider the price of stability.

\subsection{Price of Stability}

We first need the following technical lemma that bounds the coefficients $\alpha_T$, $T \subseteq N$ that occur when writing the potential function as in Corollary~\ref{cor:alpha_T}.

\label{subsubsec_shapley_set-dep_PoS}
 \begin{lemma}
\label{lem_sum_alphaT_Hk}
Let $P \in \P$ and $r \in R$. Then, $\sum_{T \subseteq P^r}{\alpha_T} = H_k$, where $k = |P^r|$ and $H_{k}$ is the $k$-th harmonic number.
\end{lemma}

\begin{proof}
Let us fix $P^r \subseteq N$ with $|P^r| = k$. Recall that $\alpha_T = \frac{(|P^r| - |T|)!(|T|-1)!}{|P^r|!}$ and $\alpha_{\emptyset}=0$, thus, $\alpha_T$
does not depend on the set $T$ but only on its cardinality $|T|$. In particular, $\alpha_S = \alpha_T$ for all $S,T \in 2^N$ with $|S| = |T|$.
Defining $\alpha_{l} = \alpha_T$ where $T \subset N$ with $|T| = l$ is arbitrary, we obtain $\alpha_{l} = \frac{(|P^r|-l)!(l-1)!}{|P^r|!}$ and $\alpha_0=0$. We calculate\begin{align*}
  \sum_{T \subseteq P^r} \alpha_T &= \sum_{l=1}^{k} \sum_{\substack{T\subseteq P^r\\|T|=l}}{\alpha_T} = \sum_{l=1}^{k} \sum_{\substack{T\subseteq P^r\\|T|=l}} \alpha_{l} = \sum_{l=1}^{k} \alpha_l \binom{k}{l} = \sum_{l=1}^{k} \frac{(l-1)! (k-l)!}{k!} \binom{k}{l}\\
  &= \sum_{l=1}^{k}{\frac{(l-1)!(k-l)!}{k!} \cdot \frac{k!}{l!(k-l)!}} = \sum_{l=1}^{k}{\frac{1}{l}} = H_k,\\[-3\baselineskip]
\end{align*}
\end{proof}

We obtain the following upper bound on the price of stability of $n$-player Shapley cost sharing games.

\begin{theorem}
\label{thm_pos_shapley_setDependent_arb_upper}
The price of stability for $n$-player Shapley cost sharing games with set-dependent cost functions is at most $n H_n$.
\end{theorem}
\begin{proof}
Fix a Shapley cost sharing game and a socially optimal strategy vector~$\hat{P}$.  We proceed to show that a global minimum $P$ of the potential function
given in Proposition~\ref{prop_phiExactPot} has cost no larger than $nH_n \cdot C(\hat{P})$. 
%
 %
 We first calculate
 \begin{align*}
   C(P) &= \sum_{r \in R} C^r(P^r)  \leq \sum_{r \in R}{|P^r| \cdot \left(\frac{1}{|P^r|} C^r(P^r) + \sum\limits_{T\subset P^r}{\alpha_T C^r(T)}\right)},
\end{align*}
since $\alpha_T = \frac{(|T| - 1)!(|P^r|-|T|)!}{|P^r|!} \geq 0$.
We use $\alpha_{P^r} = \frac{1}{|P^r|}$ to obtain
\begin{align}
 C(P)  &\leq \sum_{r \in R}|P^r| \cdot \left(\alpha_{P^r} C^r(P^r) + \sum\limits_{T\subset P^r}{\alpha_T C^r(T)}\right) = \sum_{r\in R} |P^r| \cdot \sum_{T\subseteq P^r} \alpha_T C^r(T)\notag\\
   &\leq n \cdot \sum\limits_{r\in R}{\sum_{T\subseteq P^r} \alpha_T C^r(T) } = n \cdot \Phi(P). \label{eq:bound_cp}
\end{align}
As $\Phi(P) \leq \Phi(\hat{P})$ and cost functions $C^r$ are non-decreasing, we obtain
\begin{align*}
C(P) &\leq n \cdot \Phi(\hat{P}) = n \cdot \sum_{r\in R} \sum_{T\subseteq \hat{P}^r} \alpha_T C^r(T) \leq n \cdot \sum_{r\in R} \sum_{T\subseteq \hat{P}^r} \alpha_T C^r(\hat{P}^r)\\
     &= \smash[t]{n \cdot \sum_{r\in R} C^r(\hat{P}^r) \sum_{T\subseteq \hat{P}^r} \alpha_T}.
\end{align*}
Using Lemma~\ref{lem_sum_alphaT_Hk}, we obtain $C(P) \leq nH_n \cdot C(\hat{P})$, as claimed.
\end{proof}

For $n$-player games with supermodular cost functions, we obtain a better upper bound for the price of stability of $n$. 
 
 \begin{theorem}
 \label{thm_pos_shapley_setDependent_supermod_upper}
  The price of stability for $n$-player Shapley cost sharing games with supermodular cost functions is at most $n$.
 \end{theorem}
 \begin{proof}
Let us again fix an arbitrary game and let us denote a socially optimal strategy vector and a potential minimum by $\hat{P}$ and $P$, respectively. Using inequality \eqref{eq:bound_cp} from 
the proof of Theorem~\ref{thm_pos_shapley_setDependent_arb_upper}, we obtain $C(P) \leq n\Phi(P) \leq n\Phi(\hat{P})$, so it suffices to show that $\Phi(\hat{P}) \leq C(\hat{P})$ for supermodular cost functions.
 
To this end, we first remove all players from the game and then add them iteratively to their strategy played in $\hat{P}$. Formally, for $i \in \{1,\dots,n\}$, let
$\hat{P}_{(i)}$ denote the (partial) strategy vector in which only the players $j \in \{1,\dots,i\}$ play their strategies $\hat{P}_j$ and all other players are removed from the game (and, thus, have costs 0). 
 
For every $i$, we get
$C(\hat{P}_{(i)}) - C(\hat{P}_{(i-1)})
= \sum_{r \in \hat{P}_i} C^r(\hat{P}_{(i-1)}^r \cup \{i\}) - C^r(\hat{P}_{(i-1)}^r)$,
where $\smash{\hat{P}_{(i-1)}^r}$ denotes the set of players using resource $r$ under strategy vector $\smash{\hat{P}_{(i-1)}}$. We write this expression as
\begin{align*}
C(\hat{P}_{(i)}) - C(\hat{P}_{(i-1)})
&= \sum\nolimits_{r \in \hat{P}_i} \sum\nolimits_{\pi \in \Pi(\hat{P}_{(i-1)}^r \cup \{i\})} \frac{C^r(\hat{P}^r \cup \{i\}) - C^r(\hat{P}^r)}{|\hat{P}_{(i-1)}^r \cup \{i\}|!}.
\intertext{As $C^r$ is supermodular, the marginal cost $C^r(\hat{P}^r \cup \{i\}) - C^r(\hat{P}^r)$ of player~$i$ does not increase when considering only those players
that appear before $i$ in the permutation $\pi$. Hence,}
C(\hat{P}_{(i)}) - C(\hat{P}_{(i-1)})
  &\geq \sum\nolimits_{r \in \hat{P}_i} \sum\nolimits_{\pi\in \Pi(\hat{P}^r_{(i-1)} \cup \{i\})} \frac{C^r(\hat{P}^r_{i,\pi} \cup \{i\}) - C^r(\hat{P}^r_{i,\pi})}{|\hat{P}^r_{(i-1)} \cup \{i\}|!}  \\
  &= \sum\nolimits_{r \in \hat{P}_i} \phi_i^r(\hat{P}^r_{(i-1)} \cup \{i\})\\
  &= C_i(\hat{P}_{(i)}) - C_i(\hat{P}_{(i-1)}) = \Phi(\hat{P}_{(i)}) - \Phi(\hat{P}_{(i-1)}),
 \end{align*}
where the last equation is due to the fact that $\Phi$ is an exact potential function. Thus, in each step, the cost increases at least as much as the potential, which
finally implies $\Phi(\hat{P}) \leq C(\hat{P})$, as claimed.
 \end{proof}
 
We obtain an even better bound on the price of stability of Shapley cost sharing games with submodular costs. This result has been obtained independently by \citet{Roughgarden14}.
 
\begin{restatable}{theorem}{thmposshapleysetDependentsubmodupper}
 \label{thm_pos_shapley_setDependent_submod_upper}
  The price of stability of $n$-player Shapley cost sharing games with submodular costs is at most $H_n$.
 \end{restatable}
 
\begin{proof}
  The proof works very similar to the one of Theorem \ref{thm_pos_shapley_setDependent_supermod_upper} with the difference that we now use
  the submodularity of the costs to show that $C(P) \leq \Phi(P)$ for the strategy vector $P$ that minimizes $\Phi$.
  
Let $\hat{P}$ be a socially optimal strategy vector and let $P$ be a strategy vector that minimizes $\Phi$.  
 The proof of Theorem \ref{thm_pos_shapley_setDependent_arb_upper} yields
 $\Phi(P) \leq \Phi(\hat{P}) \leq H_n C(\hat{P})$, so it suffices to show $C(P) \leq \Phi(P)$.
 
Analogously to the proof of Theorem~\ref{thm_pos_shapley_setDependent_supermod_upper}, we first remove all players from the game and then insert
them iteratively in order of their index. For $i \in \{1,\dots,n\}$ let $P_{(i)}$ denote the (partial) strategy vector, in which all players
$j \in \{1,\dots,i\}$ play $P_j$. As in the proof of Theorem~\ref{thm_pos_shapley_setDependent_supermod_upper}, we obtain
\begin{align*}
C(P_{(i)}) - C(P_{(i-1)})
&= \sum_{r \in P_i} \sum_{\pi \in \Pi(P_{(i-1)}^r \cup \{i\})} \frac{C^r(P^r \cup \{i\}) - C^r(P^r)}{|\hat{P}_{(i-1)}^r \cup \{i\}|!}.
\end{align*}
We use the submodularity of the cost functions to obtain $C(P_{(i)}) - C(P_{(i-1)}) \leq \Phi(P_{(i)}) - \Phi(P_{(i-1)})$, i.e., in each iteration,
the social cost function decreases no more than the potential function. This yields $C(P) \leq \Phi(P)$.
\end{proof}

 \subsection{Price of Anarchy}

We proceed to analyze the price of anarchy of Shapley cost sharing games. It turns out that we obtain a finite bound on the price of anarchy (for a fixed number of
players) only for submodular cost functions.  

\label{sec_shapley_PoA_setDep}
\begin{theorem}
\label{thm_poa_shapley_setDependent_submod_upper}
The price of anarchy of $n$-player Shapley cost sharing games with submodular costs is at most $n$.
\end{theorem}
\begin{proof}
 Let $P$ be a pure Nash equilibrium and let $\hat{P}$ be a socially optimal strategy vector. We calculate
 \begin{align*}
C_i(P)
&\leq C_i(\hat{P}_i, P_{-i}) = \sum\nolimits_{r \in \hat{P}_i} \sum\nolimits_{\pi \in \Pi(P^r \cup \{i\})}
\frac{C^r(P^r_{i,\pi} \cup \{i\}) - C^r(P^r_{i,\pi})}{|\hat{P}^r \cup \{i\}|!},
\intertext{and, by submodularity,}
C_i(P)
&\leq \sum\nolimits_{r \in \hat{P}_i} \sum\nolimits_{\pi \in \Pi(P^r \cup \{i\})}
\!\!\frac{C^r(\{i\}) - C^r(\emptyset)}{|\hat{P}_r \cup \{i\}|!}  = \sum\nolimits_{r \in \hat{P}_i} \!C^r(\{i\}) \leq C(\hat{P}).\\[-3\baselineskip]
\end{align*}
\end{proof}

\section{General Uniform Cost Sharing}
\label{sec:uniform}

In Section \ref{sec:shapley}, we showed upper bounds on the inefficiency of pure Nash equilibria
for the Shapley cost sharing protocol. In this section, we show that these upper bounds are essentially tight, even for network games with anonymous resource costs that
depend only on the cardinality of the set of the resource's users.

\subsection{Price of Stability}
It is well known \citep{Anshelevich08,ChenRV10} that no uniform cost sharing protocol can guarantee a price of stability strictly below $H_n$ for all $n$-player Shapley
cost sharing games with submodular costs. In fact, this negative result holds even for anonymous and constant costs.
We complement this result with the observation that even for anonymous and convex costs, no uniform cost sharing protocol can guarantee a price of stability strictly below $n$.

\begin{proposition}
\label{pro_uniform_supermod_lower}
For all uniform cost sharing protocols and $\epsilon \in (0,1)$ there is a $n$-player network resource allocation model with anonymous convex cost functions such that the
price of stability is at least $n-\epsilon$.
\end{proposition}

\begin{proof}
Let us fix a uniform cost sharing protocol and a player set $N = \{1,\dots,n\}$. Because of the uniformity of the cost sharing protocol, the cost
shares for a given edge are completely determined by the local information on that edge and do not depend on the structure of the
network.

 So let us consider an edge $e_1$ with the following convex cost function
\begin{align*}
C^{e_1}(P^{e_1}) = 
\begin{cases}
0, & \text{if } |P^{e_1}| \neq n,\\
n - \epsilon, & \text{if } |P^{e_1}| = n.
\end{cases}
\end{align*}

By budget-balance, $\sum_{i\in N} c_{i}^{e_1}(N) = n - \epsilon$, implying that there is a player $i \in N$ paying less than $1$ on this edge
when all users are using it.

\tikzstyle{vertex}=[circle,fill=black!25,minimum size=15pt,inner sep=0pt]
\tikzstyle{smallvertex}=[circle,fill=black,minimum size=3pt,inner sep=0pt]
\tikzstyle{selected vertex}  gebogener pfeil= [vertex, fill=red!24]
\tikzstyle{edge} = [draw,thick,->]
\tikzstyle{weight} = [font=\small]
\tikzstyle{selected edge} = [draw,line width=3pt,-,red!50]
\tikzstyle{ignored edge} = [draw,line width=4pt,-,black!20]
\begin{figure}[t]
\begin{center}
\begin{tikzpicture}[scale=0.5,auto,swap]
\node[vertex] (sA) at (0,0) {$s$};
\node[vertex] (si) at (0,-2.5) {$s_i$};
\node[smallvertex](n1) at (2,0){};
\node[vertex] (t) at (8,0) {$t$};
\path[edge] (sA) -- node[weight,above,black]{$e_{4}$} (n1);
\path[edge] (si) -- node[weight,right,black]{$e_{3}$} (n1);
\path[edge] (n1) -- node[weight,above,black]{$e_{1}$} (t);
\draw[->, thick] (si) to[out=0,in=225] (t);
\node[](label) at (5,-1.8){$e_{2}$};
\end{tikzpicture}
\vspace{-0.5cm}
\end{center}
\caption{Cost sharing game with price of stability arbitrarily close to $n$.}
\label{fig_uniform_PoS_k}
\end{figure}
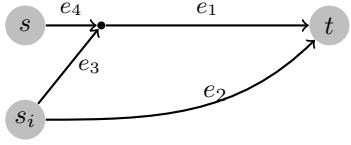

Consider the network shown in Figure~\ref{fig_uniform_PoS_k} where the cost of edge $e_2$ equals the number of its users. The other edges $e_3$
and $e_4$ are free. Player $i$ has to route from node $s_i$ to node $t$ and all the other players route from $s$ to $t$ and must use the strategy
$\{e_1,e_4\}$. If player~$i$ uses the strategy $\{e_1,e_3\}$, the costs equal $C(\{e_1,e_3\},\{e_1,e_4\},\dots,\{e_1,e_4\}) = C^{e_1}(N) = n-\epsilon$.
If she chooses $\{e_2\}$ instead, the costs are $1$. Clearly, the first strategy vector is the unique pure Nash equilibrium while the latter is the system
optimum. Thus, the price of stability is $n-\epsilon$.
\end{proof}

We proceed to show a lower bound on the price of stability of $\Theta(nH_n)$ for any uniform cost sharing protocol. To prove this, we use a result of \citet[Theorem~1]{Gopalakrishnan13}
who showed that the set of generalized weighted Shapley cost sharing protocols is the unique maximal set of uniform cost sharing protocols that is stable. To state their result,
we first recall the definition of the generalized weighted Shapley cost sharing protocols \cite{Kalai87}.


\begin{definition}[Generalized weighted Shapley cost sharing]
Let $N = \{1,\dots,n\}$ be a set of players. A tuple $w=(\lambda,\Sigma)$ with  $\lambda = (\lambda_1,\dots,\lambda_n)$ and $\Sigma = (S_1,\dots,S_m)$ is called a
\emph{weight system} if $\lambda_i > 0$  for all $i \in N$ and $\Sigma$ is a partition of $N$, i.e., $S_1 \cup  \dots \cup S_m = N$ and $S_i \cap S_j = \emptyset$ for all $i \neq j$.

Given a resource allocation model $\mathcal{M} = (N,R,\P,(C^r)_{r \in R})$ and a weight system $w = (\lambda,\Sigma)$, the \emph{generalized weighted Shapley cost sharing protocol}
assigns the cost shares $(\psi_i^r)_{i \in N, r\in R}$ defined as
 \begin{align*}
\psi_i^r(P^r) =
\begin{cases}
\sum_{T \subseteq P^r \mid i \in \overline{T}} \frac{\lambda_i}{\sum_{j \in \overline{T}} \lambda_j}\left(\sum_{U \subseteq T} (-1)^{|T|-|U|}C^r(U) \right),
& \text{ if } i \in P^r\\[10pt]
0,
& \text{ otherwise,}
\end{cases}
 \end{align*}
 where  $\overline{T} = T \cap S_k \text{ and }k = \min\{j \mid S_j\cap T\neq \emptyset \}$.
\end{definition}

The following proposition is a special case of \cite[Theorem 1]{Gopalakrishnan13}.

\begin{proposition}[\citet{Gopalakrishnan13}]
\label{pro:gopa}
For every budget-balanced, uniform and stable cost sharing protocol $\xi$ there is a weight system
 $w$ such that $\xi$ is equivalent to the generalized weighted Shapley cost sharing protocol with weight system $w$.
\end{proposition}

We use this characterization to obtain the following lower bound.

\begin{theorem}
\label{thm_uniform_arbitrary_lower}
For all uniform cost sharing protocols, $\epsilon \in \left(0,\frac{1}{2}\right)$ and $n$ even, there is a network $n$-player resource allocation model with anonymous costs for
which the price of stability is equal to $\frac{n/2+1}{1+\epsilon} H_{n/2} \in \Theta(nH_n)$.
\end{theorem}

\begin{proof}
Referring to Proposition~\ref{pro:gopa}, it is sufficient to show the claimed result for an arbitrary
generalized weighted Shapley cost sharing protocol with weight system $w$ only.

  \tikzstyle{vertex}=[circle,fill=black!25,minimum size=15pt,inner sep=0pt]
  \tikzstyle{smallvertex}=[circle,fill=black,minimum size=3pt,inner sep=0pt]
  \tikzstyle{selected vertex}  gebogener pfeil= [vertex, fill=red!24]
  \tikzstyle{edge} = [draw,thick,->]
  \tikzstyle{weight} = [font=\small]
  \tikzstyle{selected edge} = [draw,line width=3pt,-,red!50]
  \tikzstyle{ignored edge} = [draw,line width=4pt,-,black!20]

  \begin{figure}[t]
  \begin{center}
  \begin{tikzpicture}[scale=0.5,auto,swap]
    \node[vertex](sA) at (0,0){$s_A$};
    \node[smallvertex](n1) at (1,0){};
    \node[smallvertex](n2) at (2,0){};
    \node[smallvertex](n3) at (3,0){};
    \node[smallvertex](n4) at (4,0){};
    \node[smallvertex](n5) at (5,0){};
    \node[smallvertex](n6) at (6,0){};
    \node[smallvertex](n7) at (7,0){};
    \node[smallvertex](n8) at (8,0){};
    \node[vertex](tA) at (10,0){$t_A$};

    \node[vertex](sB1) at (0,-2){$s_{B_1}$};
    \node[vertex](sB2) at (2,-2){$s_{B_2}$};
    \node[vertex](sB3) at (4,-2){$s_{B_3}$};
    \node[vertex](sBk) at (6,-2){$s_{B_{\tilde{n}}}$};
    \node[vertex](tB1) at (3,2){$t_{B_1}$};
    \node[vertex](tB2) at (5,2){$t_{B_2}$};
    \node[vertex](tB3) at (7,2){$t_{B_3}$};
    \node[vertex](tBk) at (9,2){$t_{B_{\tilde{n}}}$};
 
    \node[smallvertex](unten) at (3,-3){};
    \node[smallvertex](oben) at (6,3){};


    
 
    \path[edge] (sA) -- (n1);
    \path[edge] (n1) -- node[weight,above,black]{$e_{1}$} (n2);
    \path[edge] (n2) -- (n3);
    \path[edge] (n3) -- node[weight,above,black]{$e_{2}$} (n4);
    \path[edge] (n4) -- (n5);
    \path[edge] (n5) -- node[weight,above,black]{$e_{3}$} (n6);
    \node[](dots1) at (6.5,0){\dots};
    \path[edge] (n7) -- node[weight,above,black]{$e_{\tilde{n}}$} (n8);
    \path[edge] (n8) -- (tA);
\path[edge] (sB1) -- (n1);
\path[edge] (sB2) -- (n3);
\path[edge] (sB3) -- (n5);
\path[edge] (sBk) -- (n7);
\node[](dots2) at (5.5,-1){\dots};
    
\path[edge] (sB1) -- (unten);
\path[edge] (sB2) -- (unten);
\path[edge] (sB3) -- (unten);
\path[edge] (sBk) -- (unten);
\path[edge] (n2) -- (tB1);
\path[edge] (n4) -- (tB2);
\path[edge] (n6) -- (tB3);
\path[edge] (n8) -- (tBk);
\node[](dots3) at (7.5,1){\dots};
    
\path[edge] (oben) -- (tB1);
\path[edge] (oben) -- (tB2);
\path[edge] (oben) -- (tB3);
\path[edge] (oben) -- (tBk);

\draw[->, thick] (unten) to[out=330,in=270] (12,0) to[out=90,in=30] (oben);
\node[](label) at (13,0){$e_{\tilde{n}+1}$};


\end{tikzpicture}
~\vspace{-1cm}
\end{center}
\caption{Cost sharing game with price of stability equal to
$\left(\frac{n}{2}+1\right)H_{n/2}$}
\label{fig_kHk_unweighted_example}
\end{figure}
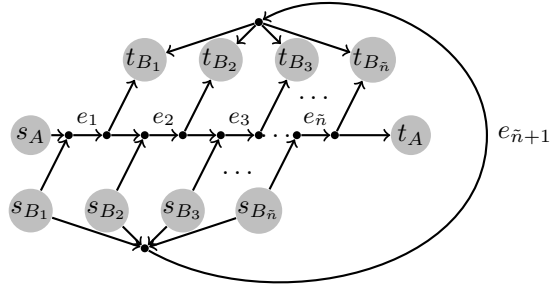

To this end, consider the network resource allocation model with anonymous costs shown in Figure \ref{fig_kHk_unweighted_example}.
   There are $n$ players that we partition into two sets $A$ and $B$ of cardinality $\tilde{n} := n/2$, i.e., $N = A \cup B$, $A\cap B = \emptyset$, and $|A|=|B|=\frac{n}{2} = \tilde{n}$. For
   edges $e_1, e_2,\dots,e_{\tilde{n}},e_{\tilde{n}+1}$ we assume the following cost functions:
\begin{align*}
C^{e_{i}}(P^{e_i}) &=
\begin{cases}
0,
& \text{if }|P^{e_i}|<\tilde{n}+1,\\
\frac{\tilde{n}+1}{i},
& \text{if }|P^{e_i}|=\tilde{n}+1,
\end{cases}
\qquad \text{ for all } i=1,\dots,\tilde{n}.\\
C^{e_{\tilde{n}+1}}(P^{e_{\tilde{n}+1}}) &=
\begin{cases}
0,
& \text{if }P^{e_{\tilde{n}+1}} = \emptyset,\\
1+\epsilon,
& \text{if }P^{e_{\tilde{n}+1}} \neq \emptyset.
\end{cases}
\end{align*}
All other edges are for free.
  
Note that the cost function of edge $e_{\tilde{n}+1}$ is constant as the cost is equal for every non-empty set of users. In previous work \cite[Lemma~5.2]{ChenRV10}
it was shown that in games where all  cost functions are constant, all uniform cost sharing protocols have to be monotone
in the sense that a player does not pay less if the cost of an edge has to be divided among less
 players, i.e.,
   $$c_{i}^{e_{\tilde{n}+1}}(P^{e_{\tilde{n}+1}}) \leq c_{i}^{e_{\tilde{n}+1}}(P^{e_{\tilde{n}+1}} \setminus \{j\}) \qquad
\forall i\neq j \in P^{e_{\tilde{n}+1}} \subseteq N.$$
This result applies to our setting as well due to the uniformity of the protocol. It is worth noting, however, that it does not necessarily hold on the other edges in
our game, since their cost is not constant.
  
  For a given generalized weighted Shapley cost sharing protocol with weight system
$w=(\lambda,\Sigma)$
 we assign the players to the sets and start nodes as follows.
 First, we sort the players ascending to their sets $S_i$ where ties are broken in favor of larger $\lambda_j$, i.e., for $a \in S_{i}$ and $b \in S_{j}$ we have the property
 \begin{align*}
  a < b \Leftrightarrow S_{i} < S_{j} \vee \left(S_{i} = S_{j} \wedge \lambda_a\geq \lambda_b \right)\text{.}
 \end{align*}. Up to
renaming we
 can assume that the obtained order is $(1,\dots,n)$.
We then put the first $n/2$ players into $A$ and the remaining $n/2$ into $B$.

All players in $A$ have start node $s_A$ and target node $t_A$. The start and target nodes of the players in $B$ are assigned in a similar fashion as in the proof of
Proposition \ref{pro_uniform_supermod_lower}. 
We first pick the player, say~$i_{\tilde{n}}$ with largest cost share on edge $e_{\tilde{n}+1}$ among the players in $B$ using $e_{\tilde{n}+1}$. We let $i_{\tilde{n}}$
route from node $s_{\tilde{n}}$ to node $t_{\tilde{n}}$. Then, we consider the remaining players
 $B \setminus \{i_{\tilde{n}}\}$ and choose a player
 \begin{align*}
  i_{\tilde{n}-1} \in \argmax\limits_{b\in B \setminus \{i_{\tilde{n}}\}}{c_{b}^{e_{\tilde{n}+1}}(B \setminus \{i_{\tilde{n}}\})}\text{.}
 \end{align*}
and let $i_{\tilde{n}-1}$ route from $s_{\tilde{n}-1}$ to $t_{\tilde{n}-1}$. We iterate this process
 until all players in $B$ are assigned to their respective start and target nodes.

 We proceed to show that in the unique pure Nash equilibrium no player uses edge $e_{\tilde{n}+1}$.
 For a contradiction, suppose $P$ is a pure Nash equilibrium and there is a non-empty set of players $\tilde{B}\subseteq B$ using edge
$e_{\tilde{n}+1}$. Let $j = \max \{k \in \{1,\dots,\tilde{n}\} : i_k \in \tilde{B}\}$ be the
 player chosen first among the players in $\tilde{B}$ in the assignment procedure described above.
 First, we show that $i_j$ pays more than $1/j$ for edge $e_{\tilde{n}+1}$ in $P$. To this end, note that
$c_{i_j}^{e_{\tilde{n}+1}}(\tilde{B}) \geq c_{i_j}^{e_{\tilde{n}+1}}(\cup_{k=1}^{j} \{i_k\}) \geq (1+\epsilon)/j > 1/j$,
 where for the first inequality we used the monotonicity of the cost shares, and for the second inequality, we used that player $i_j$ pays most among the players in $\cup_{k=1}^{j} \{i_k\}$ by construction.
 
We proceed to argue that player~$i_j \in B$ would not pay more than $1/j$ on edge $e_j$ when deviating to her other path, which is then a contradiction to
$P$ being a pure Nash equilibrium with a non-empty set of players using edge $e_{\tilde{n}+1}$.
 We calculate the cost share of player $i_j$ on $e_j$ by the definition of the generalized weighted Shapley cost sharing
protocol.
 \begin{align*}
  c_{i_j}^{e_j}(A \cup \{i_j\}) &= \sum\nolimits_{T\subseteq A \cup \{i_j\}:i_j \in \overline{T}}{\frac{\lambda_{i_j}}{\sum\nolimits_{k \in \overline{T}}{\lambda_k}}\left(\sum\nolimits_{R \subseteq T}{(-1)^{|T|-|R|}C^{e_j}(R)}\right)}\\
     &= \sum\nolimits_{T= A \cup \{i_j\}:i_j \in \overline{T}}{\frac{\lambda_{i_j}}{\sum\nolimits_{k \in
\overline{T}}{\lambda_k}}\left(C^{e_j}(T)\right)},\\
\intertext{where we used that $C^{e_i}(R)=0$ for $R \neq A \cup \{i_j\}$. We obtain}
  c_{i_j}^{e_j}(A \cup \{i_j\}) &=
\begin{cases}   
\frac{\lambda_{i_j}}{\sum_{k \in {S_1}}\lambda_k} \left(C^{e_j}(A \cup \{i_j\})\right),
&\text{if $i_j \in S_1$},\\
0,
& \text{otherwise.}
\end{cases}
\end{align*}
In particular, only the players in $S_1$ have to pay for edge $e_j$. If $i_j  \notin S_1$ we immediately obtain 
$c_{i_j}^{e_j}(A \cup \{i_j\}) = 0 \leq 1/j\text{.}$ If, on the other hand, player $i_j \in S_1$ we know from the construction of $A$ and $B$ that $A \subseteq S_1$.
Player $i_j$ has the smallest weight of all these players by construction, so $i_j$ pays not more than all the other players. We then obtain
$\smash{c_{i_j}^{e_j}(A \cup \{i_j\}) \leq \frac{(\tilde{n}+1)/j}{\tilde{n}+1} = 1/j}$.

 We have shown that in no pure Nash equilibrium there is
 a player using edge $e_{\tilde{n}+1}$.
 This implies that in all pure Nash equilibria, each player $i_j$ uses edge $e_{j}$, so the players in $N$ pay
in total $\sum_{j=1}^{\tilde{n}}{C^{e_j}(A \cup i_j)} = \sum_{j=1}^{\tilde{n}}{\frac{\tilde{n}+1}{j}} = (\tilde{n}+1)H_{\tilde{n}}$.
The price of stability thus amounts to $\frac{1}{1+\epsilon}(n/2+1)H_{n/2}$, as claimed.
\end{proof}

\subsection{Price of Anarchy}

Von Falkenhausen and Harks~\cite{Falkenhausen13} showed that any uniform cost sharing protocol
leads to unbounded price of anarchy for cost sharing games with three (weighted) players and three parallel arcs.

We complement their result by showing that also for anonymous costs but more complicated networks, no constant price of anarchy can be obtained.
 
\begin{restatable}{theorem}{thmUniformPoaSupermodLower}
\label{thm_uniform_poa_supermod_lower}
For all uniform cost sharing protocols there is a $2$-player network resource allocation model with anonymous convex cost functions such that the price of anarchy is unbounded.
\end{restatable}

\tikzstyle{vertex}=[circle,fill=black!25,minimum size=15pt,inner sep=0pt]
\tikzstyle{smallvertex}=[circle,fill=black,minimum size=3pt,inner sep=0pt]
\tikzstyle{selected vertex}  gebogener pfeil= [vertex, fill=red!24]
\tikzstyle{edge} = [draw,thick,->]
\tikzstyle{weight} = [font=\small]
\tikzstyle{selected edge} = [draw,line width=3pt,-,red!50]
\tikzstyle{ignored edge} = [draw,line width=4pt,-,black!20]
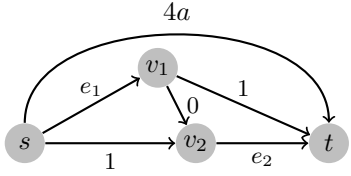
\begin{figure}[t]
\begin{center}
\begin{tikzpicture}[scale=0.5,auto,swap]
    \node[vertex] (s) at (0,0) {$s$};
    \node[vertex](n1) at (3.5,2){$v_1$};    
    \node[vertex](n2) at (4.5,0){$v_2$};
    \node[vertex] (t) at (8,0) {$t$};
    \path[edge] (s) -- node[weight,above,black]{$e_{1}$} (n1);
    \path[edge] (n1) -- node[weight,right,black]{$0$} (n2);
    \path[edge] (s) -- node[weight,below,black]{$1$} (n2);
    \path[edge] (n1) -- node[weight,above,black]{$1$} (t);
    \path[edge] (n2) -- node[weight,below,black]{$e_{2}$} (t);
    \draw[->, thick] (s) to[out=90,in=90] (t);
    \node[](label) at (4,3.5){$4a$};
 \end{tikzpicture}
 \caption{Cost sharing game used in case 1 of the proof of Theorem
\ref{thm_uniform_poa_supermod_lower}.}
 \label{fig_inf_firstCase}
 \end{center}
  \end{figure}

\begin{proof}
 Let us fix a uniform cost sharing protocol and let $a \geq 1$ be arbitrary. In the following, we will show that there is a $2$-player network resource allocation model with anonymous
 convex cost such that the resulting price of anarchy is at least $a$. As $a$ was chosen arbitrarily, this implies the claimed result.
 
First, consider the two-player game shown in Figure \ref{fig_inf_firstCase}.
 For arbitrary $q \geq 2$, we assume the following supermodular cost functions for the edges $e_1$ and $e_2$.
 \begin{align*}
   C^{e_1}(\emptyset) &= 0 & C^{e_2}(\emptyset) &= 0\\
   C^{e_1}(\{1\}) &= 1 & C^{e_2}(\{1\}) &= 1\\
   C^{e_1}(\{2\}) &= 1 & C^{e_2}(\{2\}) &= 1\\
   C^{e_1}(\{1,2\}) &= q & C^{e_2}(\{1,2\}) &= q\text{.}
  \end{align*}
  The per-player costs of all the other edges are given in Figure \ref{fig_inf_firstCase}. Both players have to
  route from $s$ to $t$. Observe that $\min_{i,j\in \{1,2\}}{\left\{c_{i}^{e_j}\left(\{1,2\}\right)\right\}}$ may depend on $q$ as the costs of edges $e_1$ and $e_2$ when used by both players are $q$.
  We distinguish the following two cases.

\textbf{First case:} $\min_{i,j\in \{1,2\}}{\left\{c_{i}^{e_j}\left(\{1,2\}\right)\right\}}$ is unbounded for $q~\in~[2,\infty)$. Fix $q$ such that
  $$\min\limits_{i,j\in \{1,2\}}{\left\{c_{i}^{e_j}\left(\{1,2\}\right)\right\}} \geq 4a > 2\text{.}$$
 This implies that if both players choose the direct edge, one of them has to pay at least $2$ for this edge. Up to labeling let this be player $2$.
  In a socially optimal strategy vector, one player routes from $s$ via $v_1$ to $t$ and the other player from $s$ via $v_2$ to $t$. Because of budget-balance each player pays $2$ for her strategy, so
  this strategy vector has a total cost of $4$.
  
  However, player~$1$ using the direct edge from $s$ to $t$ and player~$2$ using
  the path $s \rightarrow v_1 \rightarrow v_2 \rightarrow t$ is a pure Nash equilibrium.
  Player $1$ cannot improve her strategy because she can not pay less on any of the paths
  containing edge $e_1$ or edge $e_2$, because of the choice of the cost function of the direct edge.
  Player~$2$ would not deviate to the direct edge by construction and not to the other edges due to
  budget-balance.
  So this game has a price of anarchy of at least
  $$\frac{\min\limits_{i,j\in \{1,2\}}{\left\{c_{i}^{e_j}\left(\{1,2\}\right)\right\} + 2}}{4} \geq \frac{4a+2}{4} \geq
a\text{.}$$
This finishes the proof for the first case.

\textbf{Second case:} $\min_{i,j\in \{1,2\}}\{c_{i}^{e_j}(\{1,2\})\}$ is bounded for all $q \in [2,\infty)$. Let $z \in \mathbb{N}$ be such that 
\begin{align}\label{eq:bounded_min}
\min_{i,j\in \{1,2\}}\left\{c_{i}^{e_j}\left(\{1,2\}\right)\right\} \leq z \qquad \text{ for all } q \in [2,\infty).
\end{align}
Let us fix $q\geq \min\{a(z+1),2\}$. Using \eqref{eq:bounded_min}, there is a
player
  $i \in \{1,2\}$ and an edge $e \in \{e_1, e_2\}$ with
  $$c_{i}^{e}\left(\{1,2\}\right) \leq z\text{.}$$
  Up to labeling let this be player $1$ at edge $e_1$. Now consider the network shown in Figure \ref{fig_inf_secondCase}.
  Player $2$ has to route from $s_{2}$ to $t$, so she has to use edge $e_1$. Player~$1$ starts at $s_1$, so she
  can choose the direct edge from $s_1$ to $t$ or path $s_1 \rightarrow s_2 \rightarrow t$ using $e_1$.
  In a socially optimal strategy vector she would choose the direct edge and we have
  \begin{align*}
   C\left((s_1,t),e_1\right) = C^{(s_1,t)}(\{1\}) + C^{e_1}(\{2\}) = z + 1\text{,}
  \end{align*}
  whereas we know by construction of the second case that $P=(e_1,e_1)$ is a pure Nash equilibrium.
  So we immediately get a price of anarchy of
  $$\frac{C(e_1,e_1)}{C\big((s_1,t),e_1\big)} = \frac{C^{e_1}(\{1,2\})}{C^{(s_1,t)}(\{1\}) + C^{e_1}(\{2\})} = \frac{q}{z+1}
\geq \frac{a(z+1)}{z+1} \geq a\text{.}$$
This finishes the proof of the second case.

We conclude that there is a game with anonymous convex costs and arbitrarily large price.
  \end{proof}
  
   \tikzstyle{vertex}=[circle,fill=black!25,minimum size=15pt,inner sep=0pt]
  \tikzstyle{smallvertex}=[circle,fill=black,minimum size=3pt,inner sep=0pt]
  \tikzstyle{selected vertex}  gebogener pfeil= [vertex, fill=red!24]
  \tikzstyle{edge} = [draw,thick,->]
  \tikzstyle{weight} = [font=\small]
  \tikzstyle{selected edge} = [draw,line width=3pt,-,red!50]
  \tikzstyle{ignored edge} = [draw,line width=4pt,-,black!20]
 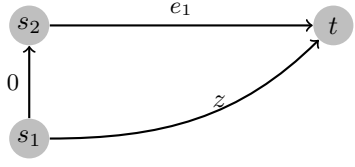
\begin{figure}[t]
 \begin{center}
  \begin{tikzpicture}[scale=0.5,auto,swap]
    \node[vertex] (s2) at (0,0) {$s_2$};
    \node[vertex](s1) at (0,-3){$s_1$};    
    \node[vertex](t) at (8,0){$t$};
    \path[edge] (s2) -- node[weight,above,black]{$e_{1}$} (t);
    \path[edge] (s1) -- node[weight,left,black]{$0$} (s2);
    \draw[->, thick] (s1) to[out=0,in=225] (t);
    \node[](label) at (5,-2){$z$};
 \end{tikzpicture}
 \end{center}
 \caption{Cost sharing game used in case 2 of the proof of Theorem
\ref{thm_uniform_poa_supermod_lower}.}
 \label{fig_inf_secondCase}
  \end{figure}
  
\newpage  
\bibliographystyle{plainnat}
\bibliography{master-bib}

\end{document}